\documentclass[12pt]{iopart}
\usepackage{iopams}
\usepackage{amsthm}
\usepackage{graphicx}
\newtheorem{Thm}{Theorem}
\newtheorem{Cor}{Corollary}
\newtheorem{Con}{Conjecture}

\newcommand{\bra}[1]{{\left\langle #1 \right|}}
\newcommand{\ket}[1]{{\left| #1 \right\rangle}}

\newcommand{\C}{\mbox{$\mathbb C$}}

\newcommand{\T}{\mbox{$\mathrm{tr}$}}

\begin{document}
\title{Monogamy and polygamy for multi-qubit entanglement using R\'enyi entropy}

\author{Jeong San Kim and Barry C. Sanders }

\address{
 Institute for Quantum Information Science,
 University of Calgary, Alberta T2N 1N4, Canada
} \eads {\mailto{jekim@ucalgary.ca}, \mailto{sandersb@ucalgary.ca} }

\date{\today}

\begin{abstract}
Using R\'enyi-$\alpha$ entropy to quantify bipartite entanglement, 
we prove monogamy of entanglement in multi-qubit systems for $\alpha
\geq 2$.
We also conjecture a polygamy inequality of multi-qubit entanglement 
with strong numerical evidence for
$0.83-\epsilon \leq \alpha \leq 1.43+\epsilon$ with $0<\epsilon<0.01$.
\end{abstract}

\maketitle

\section{Introduction}
\label{Sec: intro} 
One distinct property of quantum entanglement
from other classical correlations is its restricted sharability. For
instance, if a pair of parties in a multipartite quantum system
share maximal entanglement, then they can share neither
entanglement~\cite{ckw,ov} nor classical correlations~\cite{kw} with
the rest. This is known as the {\em Monogamy of
Entanglement}~(MoE)~\cite{T04}, and it has been shown that this
restricted sharability of quantum entanglement can be used as a
resource to distribute a secret key which is secure against
unauthorized parties~\cite{rg,m}.

Whereas MoE is a restricted property of entanglement in multipartite quantum systems,
the sharability itself is about the bipartite entanglements among the parties in multipartite systems.
In other words, it is inevitable to have a proper way of quantifying bipartite
entanglement for a good description of the monogamy nature in multipartite quantum systems.
For this reason, certain criteria of bipartite entanglement measure were recently proposed 
for a good description of the monogamy nature of entanglement in multipartite quantum systems~\cite{kds};
that is,

\begin{itemize}
\item[(i)] {\em Monotonicity}: the property that ensures
entanglement cannot be increased under local operations and classical communications.

\item[(ii)] {\em Separability}: capability
of distinguishing entanglement from separability.
\end{itemize}

\begin{itemize}
\item[(iii)] {\em Monogamy}: upper bound on a sum of bipartite entanglement
measures thereby showing that bipartite sharing of entanglement is bounded.
\end{itemize}

The first mathematical characterization of MoE was shown in
three-qubit systems~\cite{ckw} using {\em concurrence}~\cite{ww}
as the bipartite entanglement measure.
It is known as {\em CKW-inequality} named after its
establishers, Coffman, Kundu and Wootters, and 
this CKW-type inequality was also shown for arbitrary multi-qubit
systems later~\cite{ov}. In other words, concurrence is a good entanglement measure
for multi-qubit-systems that satisfies the criteria proposed in~\cite{kds}.  

However, monogamy inequality using concurrence is know to fail in
its generalization for higher-dimensional quantum systems~\cite{ou,
kds}. Furthermore, although MoE in multi-qubit systems is
mathematically well-characterized in terms of concurrence, it is not
generally true for other entanglement measures such as {\em
Entanglement of Formation}~(EoF)~\cite{bdsw}. 
In other words, MoE does
not have CKW-type characterization in terms of EoF, and this exposes
the importance of the choice of a bipartite entanglement measure to
characterize MoE even in multi-qubit systems.
Moreover, for possible generalization
of monogamy inequality into higher-dimensional quantum systems, it
is undoubtedly one of the most important and necessary tasks to have
a proper way of quantifying bipartite entanglement.

R\'enyi-$\alpha$ entropy~\cite{Renyi} is a generalization of Shannon
entropy~\cite{Shannon} in terms of the non-negative real parameter
$\alpha$, and it has been widely used in the study of quantum
information theory such as quantum entanglement and
correlations~\cite{bcehas, T2, lnp}. As a way to quantify the
uncertainty of probability distribution, R\'enyi-$\alpha$ entropy
shows its selective characters with respect to $\alpha$: The higher
values of $\alpha$, R\'enyi-$\alpha$ entropy is increasingly
determined by consideration of only the highest probability events,
whereas it increasingly weights all possible events more equally,
regardless of their probabilities for lower values of $\alpha$. For
the case when $\alpha$ tends to 1, R\'enyi-$\alpha$ entropy
converges to Shannon entropy.

Here, we show that the distinct character of R\'enyi-$\alpha$
entropy with respect to the parameter $\alpha$ also appears in
establishing monogamy inequalities of entanglement in multipartite
quantum systems. Although EoF that is based on Shannon entropy
is known to fail for usual CKW-type characterization of MoE,
R\'enyi-$\alpha$ entropy can still be shown to have CKW-type
monogamy inequality for all case of $\alpha$ if it exceeds a certain
threshold. We first provide an analytic formula for the bipartite entanglement 
measure based on R\'enyi-$\alpha$ entropy
namely {\em R\'enyi-$\alpha$ entanglement} in two-qubit systems for $\alpha\geq1$,
 and we show that multi-qubit entanglement shows the usual
CKW-type monogamy inequalities in terms of R\'enyi-$\alpha$
entanglement for $\alpha \geq2$. For $\alpha
<2$, we conjecture with strong numerical evidence that
R\'enyi-$\alpha$ entropy can provide a possible dual monogamy or
{\em polygamy} inequality of multi-qubit entanglement for
$0.83-\epsilon \leq \alpha \leq 1.43+\epsilon$ with $0<\epsilon<0.01$

This paper is organized as follows. In Section~\ref{Subsec: def}, we
recall the definition of R\'enyi-$\alpha$ entropy, and
R\'enyi-$\alpha$ entanglement for bipartite quantum states. In
Section~\ref{Subsec: 2formula}, we provide an analytic formula of
R\'enyi-$\alpha$ entanglement for arbitrary two-qubit states. In
Section~\ref{Subsec: mono}, we derive a monogamy inequality of
multi-qubit entanglement in terms of R\'enyi-$\alpha$ entanglement
for $\alpha \geq 2$, and we conjecture a polygamy inequality of
multi-qubit entanglement using R\'enyi-$\alpha$ entropy with strong
numerical evidences in Section~\ref{Subsec: poly}. Finally, we
summarize our results in Section~\ref{Conclusion}.


\section{R\'enyi-$\alpha$ Entropy and Entanglement Measures}
\label{Sec: Measure}
\subsection{Definition}
\label{Subsec: def} For a probability distribution $P=\{p_i\}$ where
$0\leq p_i\leq1$ for all $i$ and $\sum_{i}p_{i}=1$, its classical
R\'enyi-$\alpha$ entropy is defined as
\begin{equation}
H_{\alpha}(P)=\frac{1}{1-\alpha}\log\sum_{i}p_i^{\alpha},
\label{cr-entropy}
\end{equation}
for any positive $\alpha$ such that $\alpha \neq 1$. Throughout this
paper, the logarithmic function is assumed to have base two
otherwise specified. In the limiting case where $\alpha$  tends to 1,
$H_{\alpha}(P)$ converges to Shannon
entropy, that is,
\begin{equation}
\lim_{\alpha \rightarrow1}H_{\alpha}(P)=-\sum_{i}p_i\log p_i
=H(P).
\label{renyi_1}
\end{equation}

For any quantum state $\rho$, its quantum R\'enyi-$\alpha$ entropy
is defined as
\begin{equation}
S_{\alpha}(\rho)=\frac{1}{1-\alpha}\log \T \rho^{\alpha},
\label{qr-entropy}
\end{equation}
for any $\alpha >0$ and $\alpha \neq 1$~\cite{horo}.
For the quantum state $\rho$ with its spectral
decomposition $\rho=\sum_{i}\lambda_i\ket{\psi_i}\bra{\psi_i}$, we
have
\begin{equation}
S_{\alpha}(\rho)=H_{\alpha}(X),
\end{equation}
where $X=\{\lambda_i\}$ is the spectrum of $\rho$, and thus,
$S_{\alpha}(\rho)$ converges to the von Neumann entropy of $\rho$
for the case when $\alpha \rightarrow 1$.
In other words, Shannon entropy and von Neumann entropy are the
singular points of classical and quantum R\'enyi entropies
respectively, and those singularities are removable. For this reason, we
will just consider $H_1\left(P\right)=H\left(P\right)$ and
$S_1(\rho)=S(\rho)$ for any probability distribution $P$ and quantum
state $\rho$.

For a bipartite pure state $\ket{\psi}_{AB}$, the von Neumann
entropy of the reduced density matrix $\rho_A=\T_B
\ket{\psi}_{AB}\bra{\psi}$ is known to be a good bipartite
entanglement measure
\begin{equation}
E(\ket{\psi}_{AB})=S(\rho_A)=S(\rho_B).
\label{EoF pure}
\end{equation}
With noticing that von Neumann entropy quantifies the uncertainty of
the quantum state, this way of quantifying bipartite entanglement is
based on the uncertainty of subsystem: More uncertainty on
subsystems implies stronger quantum correlation between subsystems.

A well-known way to generalize this concept of entanglement measure into
mixed states is taking the minimum (or infimum) of the average
entanglements
\begin{equation}
E_{\rm f}(\rho_{AB})=\min \sum_i p_i E(\ket{\psi}_{AB}) \label{EoF}
\end{equation}
over all possible pure state decompositions of the mixed state
$\rho_{AB}=\sum_{i}p_i \ket{\psi_i}_{AB}\bra{\psi_i}$. This
generalization is known as {\em convex-roof extension}, and
$E_{\rm f}(\rho_{AB})$ is called the entanglement of formation of
$\rho_{AB}$.

As a generalization of EoF into the full spectrum of R\'enyi-$\alpha$ entropy~\cite{v},
{\em R\'enyi-$\alpha$ entanglement} of a bipartite pure state $\ket{\psi}_{AB}$
is defined as
\begin{equation}
E_{\alpha}\left(\ket{\psi}_{AB} \right)=S_{\alpha}(\rho_A),
\label{Eq: Renyi measure pure}
\end{equation}
where $\rho_A=\T_B \ket{\psi}_{AB}\bra{\psi}$,
and for a mixed state $\rho_{AB}$, its R\'enyi-$\alpha$ entanglement
is defined as,
\begin{equation}
E_{\alpha}\left(\rho_{AB} \right)=\min \sum_i p_i E_{\alpha}(\ket{\psi_i}_{AB}),
\label{Eq: Renyi measure pure}
\end{equation}
where the minimum is taken over all possible pure state
decompositions of $\rho_{AB}=\sum_{i}p_i
\ket{\psi_i}_{AB}\bra{\psi_i}$. 
Similar to EoF for bipartite quantum states, R\'enyi-$\alpha$ entanglement
is defined based on the uncertainty of subsystems, which has EoF as a special 
case when $\alpha\rightarrow1$.

It is direct to check that $E_{\alpha}\left(\rho_{AB} \right)=0$ if and only if
$\rho_{AB}$ is a separable state, and furthermore, 
R\'enyi-$\alpha$ entanglement is also known as entanglement monotone:
it is not increased under local operations and classical communications. 

\subsection{Analytical formula for two-qubit systems}
\label{Subsec: 2formula}

Let us recall the definition of concurrence. For any bipartite pure
state $\ket \phi_{AB}$, its concurrence,
$\mathcal{C}(\ket\phi_{AB})$ is defined as~\cite{ww}
\begin{equation}
\mathcal{C}(\ket \phi_{AB})=\sqrt{2(1-\T\rho^2_A)},
\label{pure state concurrence}
\end{equation}
where $\rho_A=\T_B(\ket \phi_{AB}\bra \phi)$,
and for any mixed state $\rho_{AB}$, its concurrence is defined as
\begin{equation}
\mathcal{C}(\rho_{AB})=\min \sum_k p_k \mathcal{C}({\ket {\phi_k}}_{AB}),
\label{mixed state concurrence}
\end{equation}
where the minimum is taken over all possible pure state
decompositions, $\rho_{AB}=\sum_kp_k{\ket {\phi_k}}_{AB}\bra
{\phi_k}$.

For any two-qubit mixed state $\rho_{AB}$ in $\mathcal B
\left(\C^{2}\otimes \C^{2}\right)$, its concurrence is known to have
an analytic formula~\cite{ww}, that is,
\begin{equation}
\mathcal{C}(\rho_{AB})=\max\{0, \lambda_1-\lambda_2-\lambda_3-\lambda_4\},
\label{C_formula}
\end{equation}
where $\lambda_i$'s are the eigenvalues, in decreasing order, of
$\sqrt{\sqrt{\rho_{AB}}\tilde{\rho}_{AB}\sqrt{\rho_{AB}}}$ and
$\tilde{\rho}_{AB}=\sigma_y \otimes\sigma_y
\rho^*_{AB}\sigma_y\otimes\sigma_y$ with the Pauli operator
$\sigma_y$. Furthermore, the relation between concurrence
and EoF of a two-qubit mixed state $\rho_{AB}$ (or a pure state
$\ket{\psi}_{AB} \in \mathbb{C}^2 \otimes \mathbb{C}^{d}$),
can be given as a monotone
increasing, convex function $\mathcal{E}$~\cite{ww}, such that
\begin{equation}
 E_{\rm f} (\rho_{AB}) = {\mathcal E}(\mathcal{C}_{AB}),
\end{equation}
where
\begin{equation} {\mathcal E}(x) = H\Bigl({1\over 2} + {1\over
2}\sqrt{1-x^2}\Bigr), \hspace{0.5cm}\mbox{for } 0 \le x \le 1,
\label{eps}
\end{equation}
with the binary entropy function $H(t) = -t\log t -
(1-t)\log (1-t)$. In other words, the analytic formula of
concurrence as well as its functional relation with EoF lead us to
an analytic formula of EoF for two-qubit states.

For any two-qubit pure state (or any pure state with Schmidt-rank
less than or equal to two) with its Schmidt decomposition
$\ket{\psi}_{AB}=\sqrt{\lambda_{0}}\ket{00}_{AB}+\sqrt{\lambda_{1}}\ket{11}_{AB}$,
its R\'enyi-$\alpha$ entanglement is
\begin{eqnarray}
E_{\alpha}\left(\ket{\psi}_{AB} \right)&=&S_{\alpha}(\rho_A)\nonumber\\
&=&\frac{1}{1-\alpha}\log\left(\lambda_0^{\alpha}+\lambda_1^{\alpha} \right).
\label{RM pure}
\end{eqnarray}

Now, for each $\alpha > 0$, by defining an analytic function
\begin{equation}
f_{\alpha}(x):=\frac{1}{1-\alpha}\log\left[\left(\frac{1-\sqrt{1-x^2}}{2}\right)^{\alpha}
+\left(\frac{1+\sqrt{1-x^2}}{2}\right)^{\alpha}\right]
\label{f_alpha}
\end{equation}
on $0 \leq x \leq 1$, it can be directly checked that
\begin{equation}
E_{\alpha}\left(\ket{\psi}_{AB} \right)=f_{\alpha}\left(\mathcal{C}(\ket \psi_{AB}) \right),
\label{relation_pure}
\end{equation}
where $\mathcal{C}(\ket \psi_{AB})$ is the concurrence of $\ket{\psi}_{AB}$.
Thus, for each $\alpha >0$, we have a
functional relation between the concurrence and R\'enyi-$\alpha$
entanglement for two-qubit pure state $\ket{\psi}_{AB}$. Here, we note that for the case when
$\alpha$ tends to 1, $f_{\alpha}(x)$ converges to the function ${\mathcal E}(x)$
in Eq.~(\ref{eps}); that is,
\begin{equation}
\lim_{\alpha -> 1}f_{\alpha}(x) = \mathcal{E}(x).
\label{eps2}
\end{equation}

For a mixed state $\rho_{AB}$, we have the following theorem.

\begin{Thm}
For each $\alpha >0$, if
\begin{equation}
f_{\alpha}(x)=\frac{1}{1-\alpha}\log\left[\left(\frac{1-\sqrt{1-x^2}}{2}\right)^{\alpha}
+\left(\frac{1+\sqrt{1-x^2}}{2}\right)^{\alpha}\right]
\end{equation}
is a monotonically increasing and convex function, then
\begin{equation}
E_{\alpha}\left(\rho_{AB} \right)=f_{\alpha}\left(\mathcal{C}(\rho_{AB}) \right),
\label{2formula}
\end{equation}
for any two-qubit state  $\rho_{AB}$ where $\mathcal{C}(\rho_{AB})$
is the concurrence and $E_{\alpha}\left(\rho_{AB} \right)$ is the
R\'enyi-$\alpha$ entanglement of $\rho_{AB}$. \label{Thm: 2-analytic}
\end{Thm}
\begin{proof}
Suppose $\rho_{AB}=\sum_i p_i\ket{\psi_i}_{AB}\bra{\psi_i}$ is an
optimal decomposition for $E_{\alpha}\left(\rho_{AB} \right)$ such
that $E_{\alpha}\left(\rho_{AB} \right)=\sum_i
p_iE_{\alpha}\left(\ket{\psi_i}_{AB}\right)$, then
\begin{eqnarray}
E_{\alpha}\left(\rho_{AB} \right)&=&\sum_i p_iE_{\alpha}\left(\ket{\psi_i}_{AB}\right)\nonumber\\
&=&\sum_i p_i f_{\alpha}\left( \mathcal{C}\left(\ket{\psi_i}_{AB}\right)\right)\nonumber\\
&\geq& f_{\alpha}\left(\sum_i p_i  \mathcal{C}\left(\ket{\psi_i}_{AB}\right)\right)\nonumber\\
&\geq& f_{\alpha}\left( \mathcal{C}\left(\rho_{AB}\right)\right)
\label{lineq}
\end{eqnarray}
where the second equality is by the relation between concurrence
and R\'enyi-$\alpha$ entanglement for pure states in
Eq.~(\ref{relation_pure}), the first inequality is by the convexity
of $f_{\alpha}$, and the last inequality is by the monotonicity of
$f_{\alpha}$ and the definition of
$\mathcal{C}\left(\rho_{AB}\right)$.

Due to the analytic formula of concurrence for two-qubit
states~\cite{ww}, we can always assume the existence of an optimal
decomposition $\rho_{AB}=\sum_i p_i\ket{\psi_i}_{AB}\bra{\psi_i}$
such that
\begin{equation}
\mathcal{C}\left(\rho_{AB}\right)=\sum_i p_i \mathcal{C}\left(\ket{\psi_i}_{AB}\right)\nonumber\\
\end{equation}
and
\begin{equation}
\mathcal{C}\left(\ket{\psi_i}_{AB}\right)=\mathcal{C}\left(\rho_{AB}\right)
\label{C_equal}
\end{equation}
for all $i$. Now, we have
\begin{eqnarray}
f_{\alpha}\left( \mathcal{C}\left(\rho_{AB}\right)\right)&=&
f_{\alpha}\left( \sum_i p_i  \mathcal{C}\left(\ket{\psi_i}_{AB}\right)\right)\nonumber\\
&=&\sum_i p_i f_{\alpha}\left(\mathcal{C}\left(\ket{\psi_i}_{AB}\right)\right)\nonumber\\
&=&\sum_i p_i E_{\alpha}\left(\ket{\psi_i}_{AB}\right)\nonumber\\
&\geq& E_{\alpha}\left(\rho_{AB}\right)
\label{rineq}
\end{eqnarray}
where the second equality is by Eq.~(\ref{C_equal}) and the
inequality is by the definition of
$E_{\alpha}\left(\rho_{AB}\right)$.

Thus, by Eqs.~(\ref{lineq}) and (\ref{rineq}), we have
\begin{equation}
E_{\alpha}\left(\rho_{AB} \right)=f_{\alpha}\left(\mathcal{C}(\rho_{AB}) \right),
\end{equation}
which completes the proof.
\end{proof}

Theorem~\ref{Thm: 2-analytic} implies that for any positive $\alpha$
such that $f_{\alpha}(x)$ in Eq.~(\ref{f_alpha}) is monotonically
increasing and convex, the analytic formula of concurrence for a
two-qubit state $\rho_{AB}$ in Eq.~(\ref{C_formula}) together with
the functional relation between $\mathcal{C}\left(\rho_{AB}\right)$
and $E_{\alpha}\left(\rho_{AB}\right)$ in Eq.~(\ref{2formula})
provide us an analytic formula for R\'enyi-$\alpha$ entanglement of
$\rho_{AB}$. The rest of this section is mainly contributed to the
analytic proof for the range of $\alpha$ where $f_{\alpha}$ is
monotonically increasing and convex.

\begin{Thm}
For any real $\alpha\geq 1$,
\begin{equation}
f_{\alpha}(x)=\frac{1}{1-\alpha}\log\left[\left(\frac{1-\sqrt{1-x^2}}{2}\right)^{\alpha}
+\left(\frac{1+\sqrt{1-x^2}}{2}\right)^{\alpha}\right]
\end{equation}
is a monotonically increasing and convex function for $0\leq x \leq
1$. Furthermore, the monotonicity and convexity of $f_{\alpha}(x)$
are strict in the sense that $f_{\alpha}(x)<f_{\alpha}(x')$ and
$f_{\alpha}(\lambda x +(1-\lambda)y) < \lambda
f_{\alpha}(x)+(1-\lambda)f_{\alpha}(y)$ for any $0<\lambda<1$ and
$0\leq x, x', y \leq 1$ such that $x < x'$. \label{Thm: int_alpha}
\end{Thm}
\begin{proof}
For the case when $\alpha$ tends to 1, $f_{\alpha}(x)$ converges to
$\mathcal{E}(x)$ in Eq.~(\ref{eps}) which is already known to be
monotonically increasing and convex~\cite{ww}. Here, we consider the
case where $\alpha>1$.

Let us define a function
\begin{equation}
g_{\alpha}(x):=-\log\left[\left(1-\sqrt{1-x^2}\right)^{\alpha}+\left(1+\sqrt{1-x^2}\right)^{\alpha}\right],
\label{g}
\end{equation}
then we have
\begin{equation}
f_{\alpha}(x)=\frac{1}{\alpha-1}g_{\alpha}(x)+\frac{\alpha}{\alpha-1},
\label{fg}
\end{equation}
which implies that the convexity and the monotonicity of
$f_{\alpha}(x)$ follow from those of $g_{\alpha}(x)$ for $\alpha>1$.
Furthermore, $g_{\alpha}(x)$ is an analytic function on $0\leq x
\leq 1$, thus the convexity and monotonicity of $g_{\alpha}(x)$ can
be assured by showing its first and second derivatives are
nonnegativity.

By taking the first derivative of $g_{\alpha}(x)$, we have
\begin{equation}
\frac{{\rm d}g_{\alpha}(x)}{{\rm d}x}=\frac{\alpha
x\left[\left(1+\sqrt{1-x^2}\right)^{\alpha-1}
-\left(1-\sqrt{1-x^2}\right)^{\alpha-1}\right]}
{\sqrt{1-x^2}\left[\left(1-\sqrt{1-x^2}\right)^{\alpha}
+\left(1+\sqrt{1-x^2}\right)^{\alpha}\right]}\geq 0 \label{1deri}
\end{equation}
for all $\alpha>1$ and the equality holds only at the boundary,
that is $x=0$ or $x=1$. In other words, $g_{\alpha}(x)$ is a
strictly monotone-increasing function for $0\leq x \leq 1$.

For the second derivative of $g_{\alpha}(x)$, we have
\begin{eqnarray}
\frac{{\rm d}^2g_{\alpha}(x)}{{\rm d}x^2}=B&&\frac{\left(A_1^{\alpha-1}-A_2^{\alpha-1}\right)
\left(A_1^{\alpha}+A_2^{\alpha}\right)}{\sqrt{1-x^2}}\nonumber\\
&&+B\left[x^2\left(A_1^{\alpha-1}-A_2^{\alpha-1}\right)^2-4(\alpha-1)x^{2\alpha-2}\right]
\label{2deri}
\end{eqnarray}
where $A_1=1+\sqrt{1-x^2}$, $A_2=1-\sqrt{1-x^2}$ and
$B=\alpha/\left[\left(1-x^2\right){\left(A_1^{\alpha}+A_2^{\alpha}
\right)}^2\right]$. The binomial series for $A_1^{\alpha-1}$ and
$A_2^{\alpha-1}$ lead us to
\begin{eqnarray}
A_1^{\alpha-1}&=&{\left( 1+ \sqrt{1-x^2}\right)}^{\alpha-1}\nonumber\\
&=&1+(\alpha-1)\sqrt{1-x^2}+
\frac{(\alpha-1)(\alpha-2)}{2!}
{\left(\sqrt{1-x^2}\right)}^2+\cdots
\label{binomial1}
\end{eqnarray}
and
\begin{eqnarray}
A_2^{\alpha-1}&=&{\left(1-\sqrt{1-x^2}\right)}^{\alpha-1}\nonumber\\
&=&1-(\alpha-1)\sqrt{1-x^2}+
\frac{(\alpha-1)(\alpha-2)}{2!}
{\left(\sqrt{1-x^2}\right)}^2-\cdots,
\label{binomial2}
\end{eqnarray}
and thus
\begin{eqnarray}
A_1^{\alpha-1}-A_2^{\alpha-1}=2(\alpha-1)\sqrt{1-x^2}+C_1 \geq 2(\alpha-1)\sqrt{1-x^2},\nonumber\\
A_1^{\alpha-1}+A_2^{\alpha-1}=2+C_2 \geq 2,
\label{lower}
\end{eqnarray}
for some non-negative $C_1$ and $C_2$.
Now, let
\begin{eqnarray}
\Gamma_1=\frac{\left(A_1^{\alpha-1}-A_2^{\alpha-1}\right)\left(A_1^{\alpha}+A_2^{\alpha}\right)}{\sqrt{1-x^2}},~
\Gamma_2=x^2\left(A_1^{\alpha-1}-A_2^{\alpha-1}\right)^2,
\end{eqnarray}
then by Eq.~(\ref{lower}), we have
\begin{equation}
\Gamma_1\geq 4(\alpha-1),~ \Gamma_2\geq 4{(\alpha-1)}^2x^2(1-x^2),
\end{equation}
and
\begin{eqnarray}
\frac{{\rm d}^2g_{\alpha}(x)}{{\rm d}x^2}&=&B\left[\Gamma_1+\Gamma_2-4(\alpha-1)x^{2\alpha-2}\right]\nonumber\\
&\geq&B\left[4(\alpha-1)-4(\alpha-1)x^{2\alpha-2}+4{(\alpha-1)}^2x^2(1-x^2)\right]\nonumber\\
&\geq&0,
\end{eqnarray}
for any $\alpha>1$ where the equality holds if and only if $x=1$.
Thus, the first and second derivatives of $g_{\alpha}(x)$ are
nonnegative for $0 \leq x \leq 1$ and strictly positive for $0<x<1$,
which implies the strict monotonicity and convexity of
$f_{\alpha}(x)$ on $0 \leq x \leq 1$ for $\alpha>1$.
\end{proof}

Now, we have the following corollary, which is the first main result of this paper.
\begin{Cor}
For any $\alpha\geq1$,
a two-qubit state $\rho_{AB}$ has an analytic formula for its
R\'enyi-$\alpha$ entanglement
such that $E_{\alpha}\left(\rho_{AB}
\right)=f_{\alpha}\left(\mathcal{C}(\rho_{AB}) \right)$ where
$\mathcal{C}(\rho_{AB})$ is the concurrence of $\rho_{AB}$,
and
\begin{equation}
f_{\alpha}(x)=\frac{1}{1-\alpha}\log\left[\left(\frac{1-\sqrt{1-x^2}}{2}\right)^{\alpha}
+\left(\frac{1+\sqrt{1-x^2}}{2}\right)^{\alpha}\right].
\end{equation}
\label{Cor: real_alpha}
\end{Cor}

In fact, $\frac{{\rm d}g_{\alpha}(x)}{{\rm d}x}$ in Eq.~(\ref{1deri}) can be
easily checked to be nonpositive for $0< \alpha <1$, and together with Eq.~(\ref{fg}), we have
$\frac{{\rm d}f_{\alpha}(x)}{{\rm d}x}=\frac{1}{\alpha-1}\frac{{\rm d}g_{\alpha}(x)}{{\rm d}x}\geq0$
for $0< \alpha <1$. In other words, the function $f_{\alpha}(x)$ in
Eq.~(\ref{f_alpha}) is a monotone-increasing function on the domain
of $0 \leq x \leq 1$ for any positive $\alpha$.

For the convexity of $f_{\alpha}(x)$ with $0< \alpha <1$, we first
note that the continuity of $\frac{{\rm d}^2g_{\alpha}(x)}{{\rm d}x^2}$ in
Eq.~(\ref{2deri}) with respect to $\alpha$ assures the positivity of
$\frac{{\rm d}^2g_{\alpha}(x)}{{\rm d}x^2}$ for $\alpha$ slightly less than 1. However, for general
$\alpha$ between 0 and 1, we tried a numerical way of calculations
for the second derivative of $f_{\alpha}(x)$ which is illustrated in
Figure~\ref{convexgraph2}.

We have tested $\frac{{\rm d}^2g_{\alpha}(x)}{{\rm d}x^2}$ for various values of
$\alpha$ between 0 and 1 and observed that
$\frac{{\rm d}^2g_{\alpha}(x)}{{\rm d}x^2}$ is nonnegative for $\alpha \geq 0.83$
(Figure~\ref{convexgraph2} (a) and (b)), whereas its positivity is
violated for $\alpha=0.82$ with $x$ around 1
(Figure~\ref{convexgraph2} (c)).
In other words, $f_{\alpha}(x)$
is numerically observed to be still convex for $0.83-\epsilon < \alpha <1$ with some
positive $\epsilon$ such that $0\leq\epsilon<0.01$.
Thus, together
with Theorem~\ref{Thm: int_alpha}, we have the following conjecture.

\begin{figure}
\parbox{5cm}{
\begin{center}
\includegraphics[width=\linewidth]{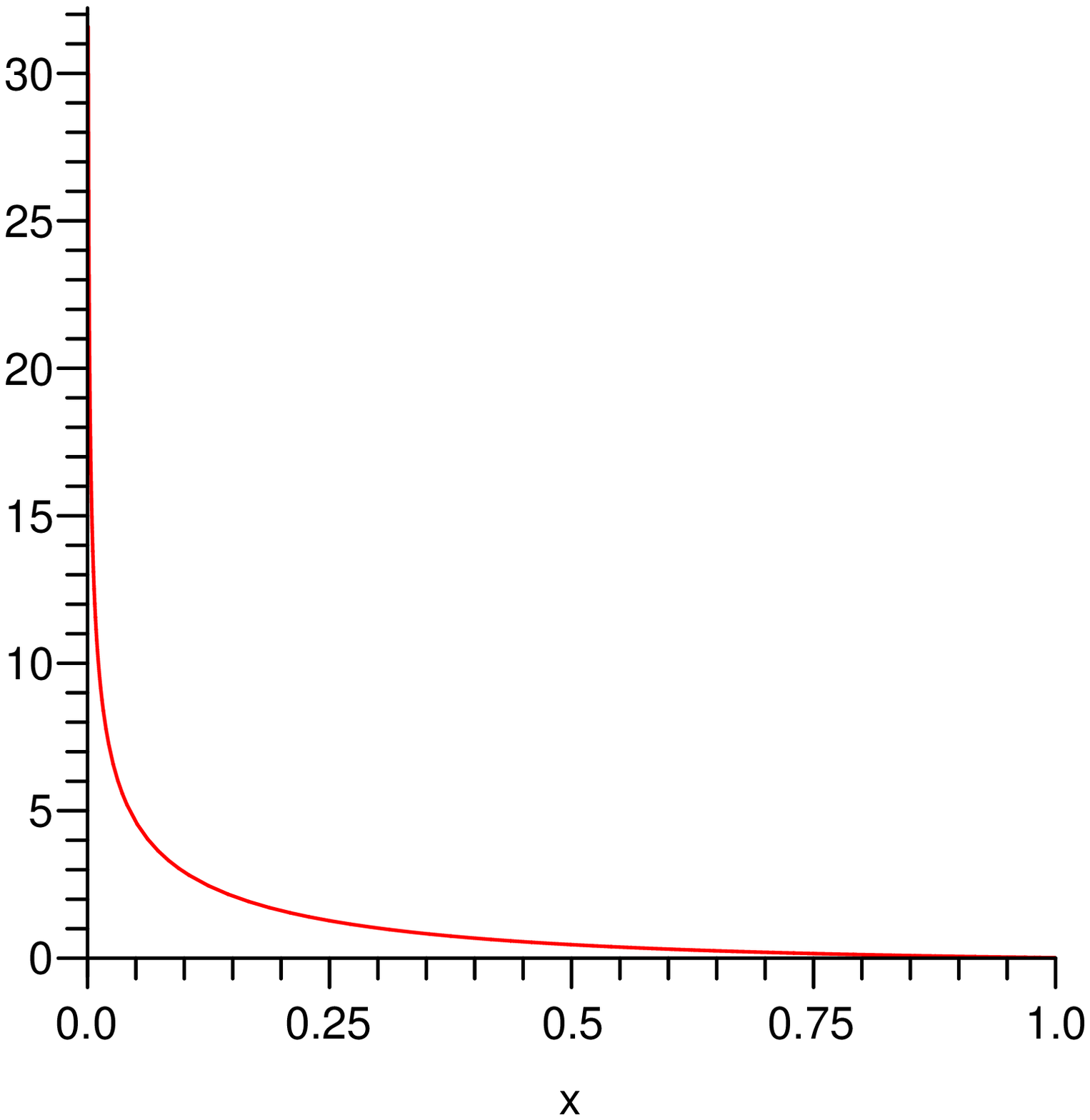}\\
(a)
\end{center}
}
\hfill
\parbox{5cm}{
\begin{center}
\includegraphics[width=\linewidth]{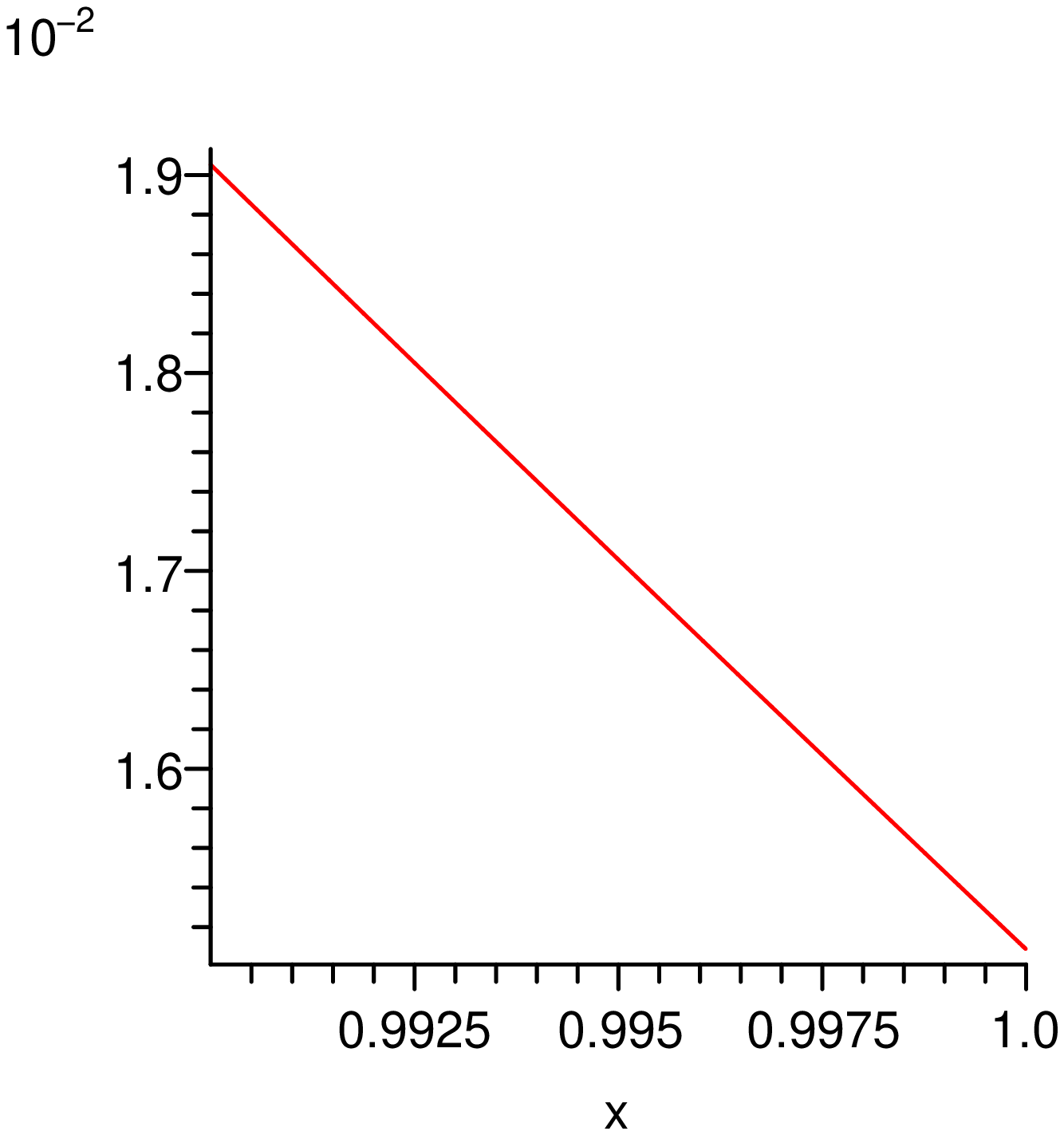}\\
(b)
\end{center}
}
\hfill
\parbox{5cm}{
\begin{center}
\includegraphics[width=\linewidth]{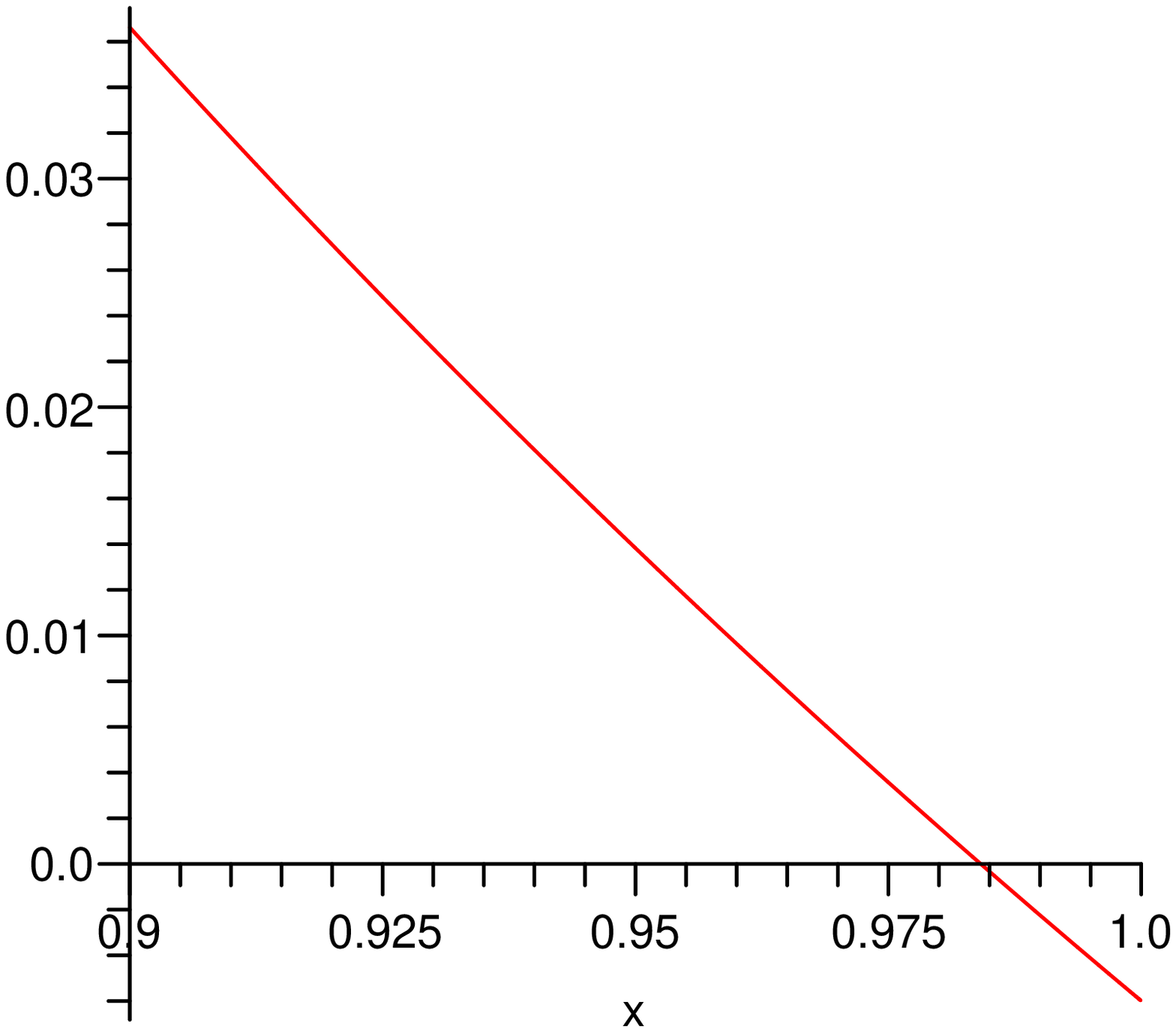}\\
(C)
\end{center}
}
\caption{$\frac{{\rm d}^2f_{\alpha}(x)}{{\rm d}x^2}$ for $\alpha=0.83$ are illustrated in (a) and (b)
on the domain $0 \leq x\leq1$ and $0.99 \leq x\leq1$ respectively.
The convexity of $f_{\alpha}(x)$ is violated for $\alpha=0.82$ in (c) where $\frac{{\rm d}^2f_{\alpha}(x)}{{\rm d}x^2}$ has
negative values for $x$ around 1. In Picture (b), the vertical axis is scaled by $10^{-2}$.
 }\label{convexgraph2}
\end{figure}

\begin{Con}
For any positive $\alpha>0.83-\epsilon$ with some positive
$\epsilon$ such that $0\leq\epsilon<0.01$,
\begin{equation}
f_{\alpha}(x)=\frac{1}{1-\alpha}\log\left[\left(\frac{1-\sqrt{1-x^2}}{2}\right)^{\alpha}
+\left(\frac{1+\sqrt{1-x^2}}{2}\right)^{\alpha}\right]
\end{equation}
is a convex function for $0\leq x \leq
1$. Furthermore, for this range of $\alpha$, any two-qubit state
$\rho_{AB}$ has an analytic formula for its R\'enyi-$\alpha$ entanglement
such that
$E_{\alpha}\left(\rho_{AB}
\right)=f_{\alpha}\left(\mathcal{C}(\rho_{AB}) \right)$ where
$\mathcal{C}(\rho_{AB})$ is the concurrence of $\rho_{AB}$.
\label{Con: alphaless1}
\end{Con}


\section{Entanglement Constraint in Multi-party Quantum Systems}
\label{Sec: monopoly}

In this section, we establish a mathematical formulation for the
monogamous and polygamous properties of multi-qubit entanglement
in terms of R\'enyi-$\alpha$ entanglement. First, we
show an important property of the function $f_{\alpha}(x)$ in
Eq.~(\ref{f_alpha}). By using the property of $f_{\alpha}(x)$ as well
as the functional relation between R\'enyi-$\alpha$ entanglement and
concurrence for two-qubit states in previous section, we derive a
monogamy inequality of multi-qubit entanglement in terms of
R\'enyi-$\alpha$ entanglement for $\alpha \geq 2$. We also conjecture a
polygamy inequality of multi-qubit entanglement in terms of
R\'enyi-$\alpha$ entanglement for $\alpha$ around 1.

\subsection{Monogamy of multi-qubit entanglement}
\label{Subsec: mono}

For a three-qubit pure state $\ket{\psi}_{ABC}$ 
CKW-inequality was shown as~\cite{ckw}, 
\begin{equation}
\mathcal{C}_{A(BC)}^2\geq\mathcal{C}_{AB}^2+\mathcal{C}_{AC}^2,
\label{eq: CKW}
\end{equation}
where
$\mathcal{C}_{A(BC)}=\mathcal{C}(\ket{\psi}_{A(BC)})$ is the
concurrence of a 3-qubit state $\ket{\psi}_{A(BC)}$ with respect to
the bipartite cut between $A$ and $BC$, and $\mathcal{C}_{AB}$ and
$\mathcal{C}_{AC}$ are the concurrences of the reduced density
matrices onto subsystems $AB$ and $AC$ respectively. Later,
Eq.~(\ref{eq: CKW}) was generalized into arbitrary multi-qubit
systems as
\begin{equation}
\mathcal{C}_{A_1 (A_2 \cdots A_n)}^2\geq\mathcal{C}_{A_1 A_2}^2 +\cdots+\mathcal{C}_{A_1 A_n}^2,
\label{eq. ncmono}
\end{equation}
for an $n$-qubit state $\rho_{A_1\cdots A_n}$~\cite{ov}.

In this section, we provide a class of monogamy inequalities in
multi-qubit systems in terms of R\'enyi-$\alpha$ entanglement.
Before we prove this, we first show a property of the function $f_{\alpha}(x)$ that plays a crucial
role in the proof of R\'enyi-$\alpha$ entanglement monogamy.

\begin{Thm}
For any real $\alpha \geq 2$ and the function,
\begin{equation}
f_{\alpha}(x)=\frac{1}{1-\alpha}\log\left[\left(\frac{1-\sqrt{1-x^2}}{2}\right)^{\alpha}
+\left(\frac{1+\sqrt{1-x^2}}{2}\right)^{\alpha}\right],
\end{equation}
we have
\begin{equation}
f_{\alpha}\left(\sqrt{x^2+y^2}\right)\geq f_{\alpha}(x)+f_{\alpha}(y)
\label{eq: int mono}
\end{equation}
for $0 \leq x,~y \leq1$ such that $0 \leq x^2+y^2 \leq 1$.
\label{Thm: mono_lemma}
\end{Thm}

\begin{proof}
Let $\mathcal{D}=\{(x,y)|~0 \leq x,~y,~x^2+y^2 \leq 1\}$ and
$h_{\alpha}(x,y)$ be a function defined on the domain $\mathcal D$ such that
\begin{equation}
h_{\alpha}(x,y):=f_{\alpha}\left(\sqrt{x^2+y^2}\right)-f_{\alpha}(x)-f_{\alpha}(y).
\label{h}
\end{equation}
Then, it is enough to show that $h_{\alpha}(x,y)$ is a non-negative
function on $\mathcal D$.

Because $h_{\alpha}(x,y)$ is continuous on the domain $\mathcal D$
and analytic in the interior of $\mathcal D$, its maximum or minimum
arises at the critical points or boundary of $\mathcal D$. Here, we
will show that $h_{\alpha}(x,y)$ has neither vanishing gradient in
the interior of $\mathcal D$ nor negative function values on the
boundary of $\mathcal D$, which implies that $h_{\alpha}(x,y)$ is a
non-negative function on the domain $\mathcal D$.

First, let us consider the gradient of $h_{\alpha}(x,y)$,
\begin{equation}
\nabla h_{\alpha}(x, y)=\left(\frac{\partial
h_{\alpha}(x,y)}{\partial x}, \frac{\partial
h_{\alpha}(x,y)}{\partial y}\right)
\label{grad}
\end{equation}
where the first-order partial derivatives of $h_{\alpha}(x,y)$ are
\begin{eqnarray}
\frac{\partial h_{\alpha}(x,y)}{\partial x}&=&
\frac{Cx\left[\left(1+\sqrt{1-x^2-y^2}\right)^{\alpha-1}-
\left(1-\sqrt{1-x^2-y^2}\right)^{\alpha-1}\right]}
{\sqrt{1-x^2-y^2}\left[\left(1-\sqrt{1-x^2-y^2}\right)^{\alpha}+
\left(1+\sqrt{1-x^2-y^2}\right)^{\alpha}\right]}\nonumber\\
&&-\frac{Cx\left[\left(1+\sqrt{1-x^2}\right)^{\alpha-1}-\left(1-\sqrt{1-x^2}\right)^{\alpha-1}\right]}
{\sqrt{1-x^2}\left[\left(1-\sqrt{1-x^2}\right)^{\alpha}+\left(1+\sqrt{1-x^2}\right)^{\alpha}\right]}
\label{deri_x}
\end{eqnarray}
and
\begin{eqnarray}
\frac{\partial h_{\alpha}(x,y)}{\partial y}&=&
\frac{Cy\left[\left(1+\sqrt{1-x^2-y^2}\right)^{\alpha-1}-
\left(1-\sqrt{1-x^2-y^2}\right)^{\alpha-1}\right]}
{\sqrt{1-x^2-y^2}\left[\left(1-\sqrt{1-x^2-y^2}\right)^{\alpha}+
\left(1+\sqrt{1-x^2-y^2}\right)^{\alpha}\right]}\nonumber\\
&&-\frac{Cy\left[\left(1+\sqrt{1-y^2}\right)^{\alpha-1}-\left(1-\sqrt{1-y^2}\right)^{\alpha-1}\right]}
{\sqrt{1-y^2}\left[\left(1-\sqrt{1-y^2}\right)^{\alpha}+\left(1+\sqrt{1-y^2}\right)^{\alpha}\right]},
\label{deri_y}
\end{eqnarray}
with $C=\frac{\alpha}{\alpha-1}$. Now, suppose $\nabla
h_{\alpha}(x_1,y_1)=(0,0)$ for some $(x_1,y_1)$ in the interior of
$\mathcal D$, that is, $0<x_1,~y_1<1$ and $0< x_1^2+y_1^2<1$.
Because both $x_1$ and $y_1$ are nonzero, from Eqs.~(\ref{deri_x})
and (\ref{deri_y}), we have
\begin{eqnarray}
\frac{\left[\left(1+\sqrt{1-{x_1}^2}\right)^{\alpha-1}-\left(1-\sqrt{1-{x_1}^2}\right)^{\alpha-1}\right]}
{\sqrt{1-{x_1}^2}\left[\left(1-\sqrt{1-{x_1}^2}\right)^{\alpha}
+\left(1+\sqrt{1-{x_1}^2}\right)^{\alpha}\right]}\nonumber\\
~~~~~~~~~~~=\frac{\left[\left(1+\sqrt{1-{y_1}^2}\right)^{\alpha-1}
-\left(1-\sqrt{1-{y_1}^2}\right)^{\alpha-1}\right]}
{\sqrt{1-{y_1}^2}\left[\left(1-\sqrt{1-{y_1}^2}\right)^{\alpha}+\left(1+\sqrt{1-{y_1}^2}\right)^{\alpha}\right]}.
\label{vaniequal}
\end{eqnarray}
Furthermore, by defining a function $l_{\alpha}(x)$ on $0< x \leq1$ such that
\begin{equation}
l_{\alpha}(x):=\frac{\left[\left(1+\sqrt{1-{x}^2}\right)^{\alpha-1}
-\left(1-\sqrt{1-{x}^2}\right)^{\alpha-1}\right]}
{\sqrt{1-{x}^2}\left[\left(1-\sqrt{1-{x}^2}\right)^{\alpha}+
\left(1+\sqrt{1-{x}^2}\right)^{\alpha}\right]},
\label{l}
\end{equation}
Eq.~(\ref{vaniequal}) can be rewritten as
\begin{equation}
l_{\alpha}(x_1)=l_{\alpha}(y_1).
\label{lequal}
\end{equation}
Here, we note that the first derivative of $l_{\alpha}(x)$ is
\begin{eqnarray}
\frac{{\rm d}l_{\alpha}(x)}{{\rm d}x}&=&B\left[\frac{\left(A_1^{\alpha}+
A_2^{\alpha}\right)\left(A_1^{\alpha-1}-A_2^{\alpha-1}\right)}{\sqrt{1-x^2}}+
\left(A_1^{\alpha-1}-A_2^{\alpha-1}\right)^2\right]\nonumber\\
&&-4B\left(\alpha-1\right)x^{2\alpha-4}
\label{l1deri}
\end{eqnarray}
with $A_1=1+\sqrt{1-x^2}$, $A_2=1-\sqrt{1-x^2}$ and
$B=\frac{x}{(1-x^2)\left(A_1^{\alpha}+A_2^{\alpha}\right)^2}$,
and from the inequalities in Eq.~(\ref{lower}), we have
\begin{eqnarray}
\frac{\left(A_1^{\alpha}+A_2^{\alpha}\right)\left(A_1^{\alpha-1}-A_2^{\alpha-1}\right)}
{\sqrt{1-x^2}}&\geq& 4(\alpha-1),\nonumber\\
\left(A_1^{\alpha-1}-A_2^{\alpha-1}\right)^2&\geq&4(\alpha-1)^2(1-x^2),
\label{l1bnd1}
\end{eqnarray}
which implies
\begin{equation}
\frac{{\rm d}l_{\alpha}(x)}{{\rm d}x}\geq 4B\left[(\alpha-1)+(\alpha-1)^2(1-x^2)-\left(\alpha-1\right)x^{2\alpha-4} \right].
\label{l1bnd2}
\end{equation}

For the region of $\alpha\geq2$ and $0< x < 1$, $\alpha-1$ in the
second term of the right-hand side of Eq.~(\ref{l1bnd2}) is always
larger than or equal to 1, whereas $x^{2\alpha-4}$ in the last term
is always less than or equal to 1, thus
\begin{eqnarray}
\frac{{\rm d}l_{\alpha}(x)}{{\rm d}x}&\geq& 4B\left[(\alpha-1)+(\alpha-1)^2(1-x^2)
-\left(\alpha-1\right)x^{2\alpha-4}\right]\nonumber\\
&\geq&4B\left[(\alpha-1)+(\alpha-1)(1-x^2)-\left(\alpha-1\right)\right]\nonumber\\
&=&4B(\alpha-1)(1-x^2)\nonumber\\
&\geq&0,
\label{l1bnd3}
\end{eqnarray}
and the last inequality is strict for $0<x<1$.
In other words,
$l_{\alpha}(x)$ is a strictly monotone-increasing function on
$0<x<1$, and therefore Eq.~(\ref{lequal}) implies $x_1=y_1$. (Also
note that the possible range of $x_1$ now becomes $0\leq x_1 \leq
\frac{1}{\sqrt{2}}$.) However, from Eqs.~(\ref{deri_x}) and
(\ref{deri_y}), $\nabla h_{\alpha}(x_1, y_1)=(0,0)$ together with
$x_1=y_1$ implies that
\begin{equation}
l_{\alpha}(\sqrt{2}x_1)=l_{\alpha}(x_1),
\end{equation}
which contradicts to the strict monotonicity of $l_{\alpha}(x)$ for
$0<x<1$. Thus, $\nabla h_{\alpha}(x_1, y_1)\neq (0,0)$ for any
$(x_1, y_1)$ in the interior of $\mathcal D$.

Now, let us consider the function value of $h_{\alpha}(x, y)$ on the
boundary of $\mathcal D$, that is $x=0$ or $y=0$, or $x^2+y^2=1$. If
one of $x$ or $y$ is 0, then it is clear to check that
$h_{\alpha}(x, y)=0$. For the case when $x^2+y^2=1$, we have
\begin{eqnarray}
h_{\alpha}(x,y)&=&f_{\alpha}\left(\sqrt{x^2+y^2}\right)-f_{\alpha}(x)-f_{\alpha}(y)\nonumber\\
&=&f_{\alpha}\left(1\right)-f_{\alpha}(x)-f_{\alpha}(\sqrt{1-x^2})\nonumber\\
&=&1-\frac{2\alpha}{\alpha-1}+\frac{1}{\alpha-1}\log\left[\left(1+x^2\right)^{\alpha}+\left(1-x^2\right)^{\alpha}\right]\nonumber\\
&&+\frac{1}{\alpha-1}\log\left[\left(1-\sqrt{1-x^2}\right)^{\alpha}+\left(1+\sqrt{1-x^2}\right)^{\alpha}\right]\nonumber\\
&=:& m_{\alpha}(x)
\end{eqnarray}
where $m_{\alpha}(x)$ is a one-parameter real-valued analytic
function defined on $0 \leq x \leq 1$. For $\alpha \geq2$, it is
also straightforward to check that the function $m_{\alpha}(x)$ has
only one critical point at $x=\frac{1}{\sqrt{2}}$ through the domain
$0 \leq x \leq 1$. Furthermore, all the function values of
$m_{\alpha}(x)$ at the critical point and the boundary are
nonnegative, that is, we have
\begin{equation}
h_{\alpha}(x,y)=m_{\alpha}(x) \geq 0
\label{m}
\end{equation}
for $x^2+y^2=1$ and $0 \leq x \leq 1$. Thus, $h_{\alpha}(x, y)$ is a
non-negative function on the domain $\mathcal D$, and this completes
the proof.
\end{proof}

Now, by using Theorem~\ref{Thm: mono_lemma} together with
Corollary~\ref{Cor: real_alpha} in previous section, we have the
following theorem, which is the secondary result of this paper.

\begin{Thm}
For $\alpha \geq 2$ and any multi-qubit state $\rho_{A_1A_2\cdots A_n}$ in
$\mathcal B \left({\C^{2}}^{\otimes n}\right)$, we have
\begin{equation}
E_{\alpha}\left(\rho_{A_1(A_2\cdots A_n)}\right)\geq
E_{\alpha}\left(\rho_{A_1A_2}\right)+\cdots +E_{\alpha}\left(\rho_{A_1A_n}\right),
\label{eq: n_mono}
\end{equation}
where $E_{\alpha}\left(\rho_{A_1(A_2\cdots A_n)}\right)$ is the
R\'enyi-$\alpha$ entanglement of $\rho_{A_1A_2\cdots A_n}$ with
respect to the bipartite cut between $A_1$ and the others, and
$E_{\alpha}\left(\rho_{A_1A_i}\right)$ is the R\'enyi-$\alpha$ entanglement of
the reduced density matrix $\rho_{A_1A_i}$ on two-qubit
subsystem $A_1A_i$ for $i=2, \ldots, n$. \label{Thm: mono}
\end{Thm}

\begin{proof}
From Eq.~(\ref{eq. ncmono}), we have
\begin{equation}
\mathcal{C}_{A_1 (A_2 \cdots A_n)}\geq \sqrt{\mathcal{C}_{A_1 A_2}^2 +\cdots+\mathcal{C}_{A_1 A_n}^2},
\label{eq. rootncmono}
\end{equation}
where $\mathcal{C}_{A_1 (A_2 \cdots A_n)}$ and $\mathcal{C}_{A_1
A_i}$ are the concurrences of $\rho_{A_1(A_2\cdots A_n)}$ and
$\rho_{A_1A_i}$ respectively.

Now, let $\rho_{A_1(A_2\cdots A_n)}=\sum_i p_i
\ket{\psi_i}_{A_1(A_2\cdots A_n)}\bra{\psi_i}$ be an optimal
decomposition for $E_{\alpha}\left(\rho_{A_1(A_2\cdots A_n)}\right)$
such that $E_{\alpha}\left(\rho_{A_1(A_2\cdots A_n)}\right)=\sum_i
p_i E_{\alpha}\left(\ket{\psi_i}_{A_1(A_2\cdots A_n)}\right)$.
Because each $\ket{\psi_i}_{A_1(A_2\cdots A_n)}$ in the
decomposition has Schmidt-rank less than or equal to two (they are
$2\otimes d$ pure states for $d=2^{\otimes n-1}$), its concurrence
and R\'enyi-$\alpha$ entanglement are related by the function
$f_{\alpha}(x)$ in Eq.~(\ref{2formula}), that is,
\begin{equation}
E_{\alpha}\left(\ket{\psi_i}_{A_1(A_2\cdots A_n)}\right)=
f_{\alpha}\left(\mathcal{C}(\ket{\psi_i}_{A_1(A_2\cdots A_n)})\right)
\end{equation}
for each i.
Thus, we have
\begin{eqnarray}
E_{\alpha}\left(\rho_{A_1(A_2\cdots A_n)}\right)
&=&\sum_i p_i E_{\alpha}\left(\ket{\psi_i}_{A_1(A_2\cdots A_n)}\right)\nonumber\\
&=&\sum_i p_i f_{\alpha}\left(\mathcal{C}(\ket{\psi_i}_{A_1(A_2\cdots A_n)})\right)\nonumber\\
&\geq&f_{\alpha}\left( \sum_i p_i \mathcal{C}(\ket{\psi_i}_{A_1(A_2\cdots A_n)})\right)\nonumber\\
&\geq&f_{\alpha}\left(\mathcal{C}_{A_1 (A_2 \cdots A_n)}\right),
\label{fineq}
\end{eqnarray}
where the first inequality is by the convexity of $f_{\alpha}(x)$,
and the second inequality is by the definition of concurrence and
the monotonicity of $f_{\alpha}(x)$. Furthermore, by Eq.~(\ref{eq.
rootncmono}) together with Theorem~\ref{Thm: mono_lemma}, we have
\begin{eqnarray}
f_{\alpha}\left(\mathcal{C}_{A_1 (A_2 \cdots A_n)}\right)&\geq&
f_{\alpha}\left(\sqrt{\mathcal{C}_{A_1 A_2}^2 +\cdots+\mathcal{C}_{A_1 A_n}^2}\right)\nonumber\\
&\geq&f_{\alpha}\left(\mathcal{C}_{A_1 A_2}\right)+f_{\alpha}\left(\sqrt{\mathcal{C}_{A_1 A_3}^2
+\cdots+\mathcal{C}_{A_1 A_n}^2}\right)\nonumber\\
&&~~~~~~~~~~~~~~~\vdots\nonumber\\
&\geq& f_{\alpha}\left(\mathcal{C}_{A_1 A_2}\right)+\cdots+f_{\alpha}\left(\mathcal{C}_{A_1 A_n}\right)\nonumber\\
&=& E_{\alpha}\left(\rho_{A_1A_2}\right)+\cdots +E_{\alpha}\left(\rho_{A_1A_n}\right),
\label{monoineq}
\end{eqnarray}
where the first inequality is by the monotonicity of
$f_{\alpha}(x)$, the other inequalities are by iterative use of
Theorem~\ref{Thm: mono_lemma}, and the last equality is by the
functional relation of R\'enyi-$\alpha$ entanglement and concurrence
for two-qubit states. Thus, Eqs.~(\ref{fineq}) together with
(\ref{monoineq}) complete the proof.
\end{proof}

In fact, it was recently shown that R\'enyi-$\alpha$ entanglement for
$\alpha=2$ can be used to establish a monogamy inequality for
multi-qubit systems by straightforward calculation~\cite{co}.
However, Theorem~\ref{Thm: mono} says that the monogamous property
of entanglement in multi-qubit systems can be mathematically
characterized in terms of R\'enyi-$\alpha$ entanglement for all positive
real number $\alpha$ larger than or equal to 2.

\subsection{Polygamy of multi-qubit entanglement}
\label{Subsec: poly}
In previous section, we have established the monogamy
inequalities of multi-qubit entanglement in terms of R\'enyi-$\alpha$
entanglement for all positive real number $\alpha\geq2$. Here, we
consider the case when $0<\alpha<2$, and claim another kind of
entanglement constraint in multi-qubit systems using
R\'enyi-$\alpha$ entropy.

Let us first recall the definition of {\em Entanglement of Assistance} (EoA)~\cite{cohen}
 for a bipartite state $\rho_{AB}$, that is,
\begin{equation}
E^a(\rho_{AB})=\max \sum_k p_k E\left({\ket {\psi_k}}_{AB}\right),
\label{EoA}
\end{equation}
where the maximum is taken over all possible pure state
decompositions of $\rho_{AB}=\sum_k p_k
\ket{\psi_k}_{AB}\bra{\psi_k}$.
Here, we note that EoA in Eq.~(\ref{EoA}) is clearly a mathematical
dual to EoF in Eq.~(\ref{EoF}) because one of them is the maximum
average entanglement over all possible pure state decompositions
whereas the other is the minimum. Moreover, by introducing a third party
$C$ that has the purification of
$\rho_{AB}$, $E^a(\rho_{AB})$ can also be considered as the maximum
achievable entanglement between $A$ and $B$ assisted by $C$. (This
is the reason why it is called the {\em assistance}.) In other
words, $E^a(\rho_{AB})$ is the maximal entanglement that can be {\em
distributed} between $A$ and $B$ assisted by the environment;
therefore, EoA is also physically dual to the concept of {\em
formation}.

Similar to the duality between EoF and EoA, we can also have a dual
concept to concurrence: {\em Concurrence of Assistance}
(CoA)~\cite{lve} for a bipartite state $\rho_{AB}$ is defined as
\begin{equation}
\mathcal{C}^a(\rho_{AB})=\max \sum_k p_k \mathcal{C}({\ket {\psi_k}}_{AB}),
\label{CoA}
\end{equation}
where the maximum is taken over all possible pure state
decompositions of $\rho_{AB}=\sum_k p_k
\ket{\psi_k}_{AB}\bra{\psi_k}$. Furthermore, it was shown that
there exists a different kind of entanglement constraint in multi-qubit
systems in terms of CoA~\cite{gbs}. More precisely, for any pure state
$\ket{\psi}_{A_1 \cdots A_n}$ in an $n$-qubit system,
we have
\begin{equation} \mathcal{C}_{A_1 (A_2 \cdots
A_n)}^2  \leq  (\mathcal{C}^a_{A_1 A_2})^2
+\cdots+(\mathcal{C}^a_{A_1 A_n})^2, \label{nCdual}
\end{equation}
where $\mathcal{C}_{A_1 (A_2 \cdots A_n)}$ is the concurrence of
$\ket{\psi}_{A_1 \cdots A_n}$ with respect to the bipartite cut between $A_1$
and $A_{2}\cdots A_{n}$, and $\mathcal{C}^a_{A_1 A_i}$ is the CoA of
the reduced density matrix
$\rho_{A_1A_i}$ for $i=2,\ldots ,n$.

Eq.~(\ref{nCdual}) is known as the dual monogamy or polygamy
inequality of entanglement in multi-qubit systems: Whereas
concurrence can be used to characterize the monogamy of multipartite
entanglement as in Eq.~(\ref{eq. ncmono}), its dual quantity, CoA
can also be used for the dual monogamy of multipartite
entanglement. Later, it was shown that polygamy of multi-qubit entanglement can also be
characterized in terms of EoA, that is, for any multi-qubit pure state
$\ket{\psi}_{A_1 \cdots A_n}$,
\begin{equation}
E\left( \ket{\psi}_{A_1(A_2 \cdots A_n)}\right)\leq  E^a(\rho_{A_1 A_2}) +\cdots+E^a(\rho_{A_1
A_n}),
\label{nEdual}
\end{equation}
where $E\left( \ket{\psi}_{A_1(A_2 \cdots A_n)}\right)=S\left(\rho_A
\right)$ is the entanglement of $\ket{\psi}_{A_1 \cdots A_n}$ with
respect to the bipartite cut $A_1$ and $A_{2}\cdots A_{n}$, and
$E^a(\rho_{A_1 A_i})$ is the EoA of the reduced density matrix
$\rho_{A_1 A_i}$ for $i=2,\cdots,n$~\cite{bgk}. More recently, a
general polygamy inequality of multipartite entanglement in
arbitrary dimensional quantum systems was proposed by using an
analytical upper bound of CoA~\cite{kim}.

Here, we claim that the polygamous property of multi-qubit
entanglement can also be shown by using R\'enyi-$\alpha$ entropy for
$\alpha$ around 1. In fact, we can intuitively expect this: Due to
the continuity of R\'enyi-$\alpha$ entropy with respect to $\alpha$,
we can easily expect that the inequality in Eq.~(\ref{nEdual}) would
also be true by using R\'enyi-$\alpha$ entropy, instead of
von Neumann entropy, where $\alpha$ is in a region of small
perturbation from 1. However, we conjecture, with strong numerical
evidences, a specific region of $\alpha$ where the polygamy
inequality can be obtained using R\'enyi-$\alpha$ entropy.

For any bipartite state $\rho_{AB}$, let us first define a
dual quantity to R\'enyi-$\alpha$ entanglement as
\begin{equation}
E_{\alpha}^a(\rho_{AB}):=\max \sum_k p_k E_{\alpha}\left({\ket {\psi_k}}_{AB}\right),
\label{RoA}
\end{equation}
where the maximum is taken over all possible pure state
decompositions of $\rho_{AB}=\sum_k p_k
\ket{\psi_k}_{AB}\bra{\psi_k}$, and call it {\em R\'enyi-$\alpha$
Entanglement of Assistance} (REoA).

Similar to the functional relation between concurrence and
R\'enyi-$\alpha$ entanglement for two-qubit states in
Eq.~(\ref{2formula}), the same function $f_{\alpha}(x)$ can also
relate REoA of a two-qubit state with its CoA. For a two-qubit state
$\rho_{AB}$, let $\rho_{AB}=\sum_{i} p_i
\ket{\psi_i}_{AB}\bra{\psi_i}$ be an optimal decomposition for its
CoA, that is,
\begin{equation}
\mathcal{C}^{a}\left( \rho_{AB} \right)=\sum_{i} p_i \mathcal {C}\left( \ket{\psi_i}_{AB}\right),
\label{CoAopt}
\end{equation}
where $\mathcal{C}^{a}\left( \rho_{AB} \right)$ is the CoA of
$\rho_{AB}$ defined in Eq.~(\ref{CoA}). By Theorem~\ref{Thm:
int_alpha}, $f_{\alpha}(x)$ is convex for the range of $\alpha \geq
1$, thus, for this range of $\alpha$, we have
\begin{eqnarray}
f_{\alpha}\left( \mathcal{C}^{a}\left( \rho_{AB} \right)\right)&=&f_{\alpha}
\left( \sum_{i} p_i \mathcal {C}\left( \ket{\psi_i}_{AB}\right)\right)\nonumber\\
&\leq& \sum_{i} p_i  f_{\alpha}\left(\mathcal {C}\left( \ket{\psi_i}_{AB}\right)\right)\nonumber\\
&=& \sum_{i} p_i  E_{\alpha}\left( \ket{\psi_i}_{AB}\right)\nonumber\\
&\leq&E_{\alpha}^a(\rho_{AB}),
\label{f_Ea}
\end{eqnarray}
where the first inequality is by the convexity of $f_{\alpha}(x)$
and the last inequality is by the definition of REoA. Furthermore,
based on Conjecture~\ref{Con: alphaless1}, we also claim that
Eq.~(\ref{f_Ea}) is true for $\alpha > 0.83-\epsilon$ with small
$\epsilon$.

Now, let us consider a property of $h_{\alpha}(x,y)$ in
Eq.~(\ref{h}). From Theorem~\ref{Thm: mono_lemma}, we first note
that the continuity of $h_{\alpha}(x,y)$ in Eq.~(\ref{h}) with
respect to $\alpha$ assures its positivity for $\alpha$ slightly
less than 2. In other words, the monogamy inequality in
Eq.~(\ref{eq: n_mono}) of Theorem~\ref{Thm: mono} that is
analytically proven for $\alpha \geq2$ is also true for this region
of $\alpha$. However, $h_{\alpha}(x,y)$ has negative function
values for most case of $0<\alpha <2$, which is illustrated in
Figure~\ref{hgraph2}. In Figure~\ref{hgraph2} (a), negative function
values of $h_{\alpha}(x,y)$ are observed near
$(x,y)=(0,0)$ when $\alpha=1.9$. Furthermore, the region in the
domain of $h_{\alpha}(x,y)$ with negative function values is getting
larger as $\alpha$ decreases. (Figure~\ref{hgraph2} (b))

\begin{figure}
\parbox{5.5cm}{
\begin{center}
\includegraphics[width=\linewidth]{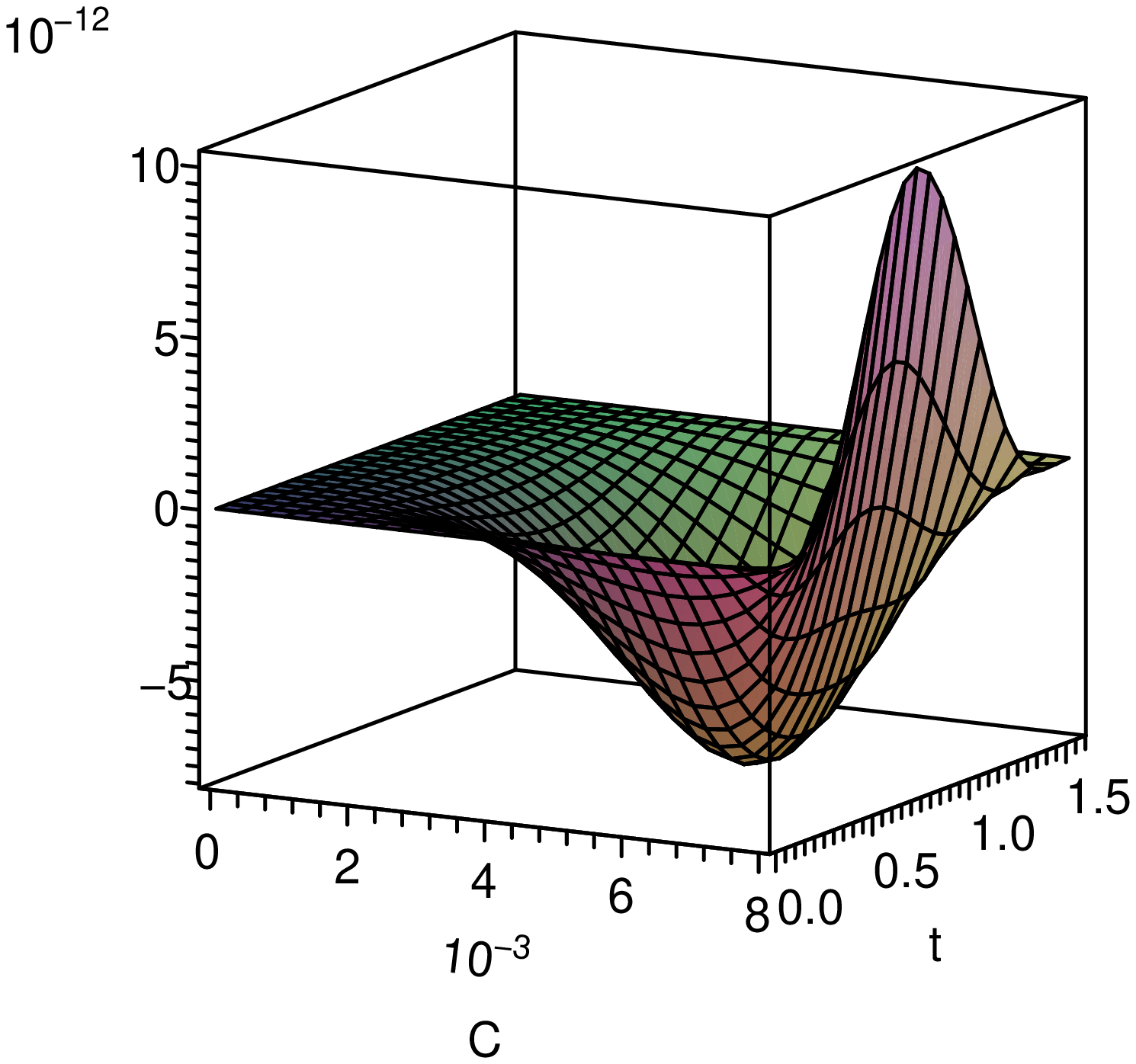}\\
(a) Graph of $h_{1.9}(x,y)$ for $x^2+y^2\leq 0.008$.
\end{center}
}
\hfill
\parbox{4.8cm}{
\begin{center}
\includegraphics[width=\linewidth]{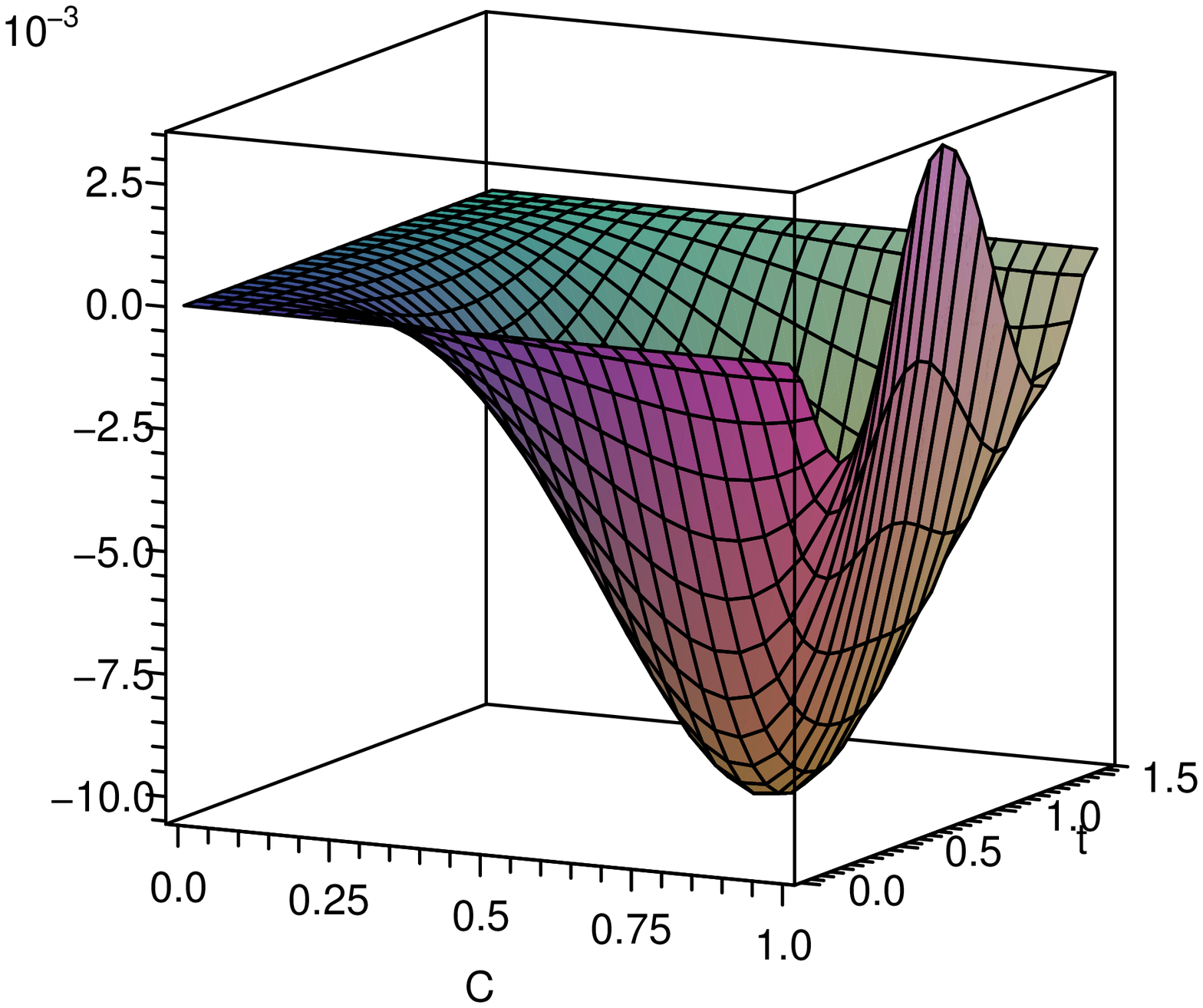}\\
(b) Graph of $h_{1.44}(x,y)$.
\end{center}
}
\hfill
\parbox{4.8cm}{
\begin{center}
\includegraphics[width=\linewidth]{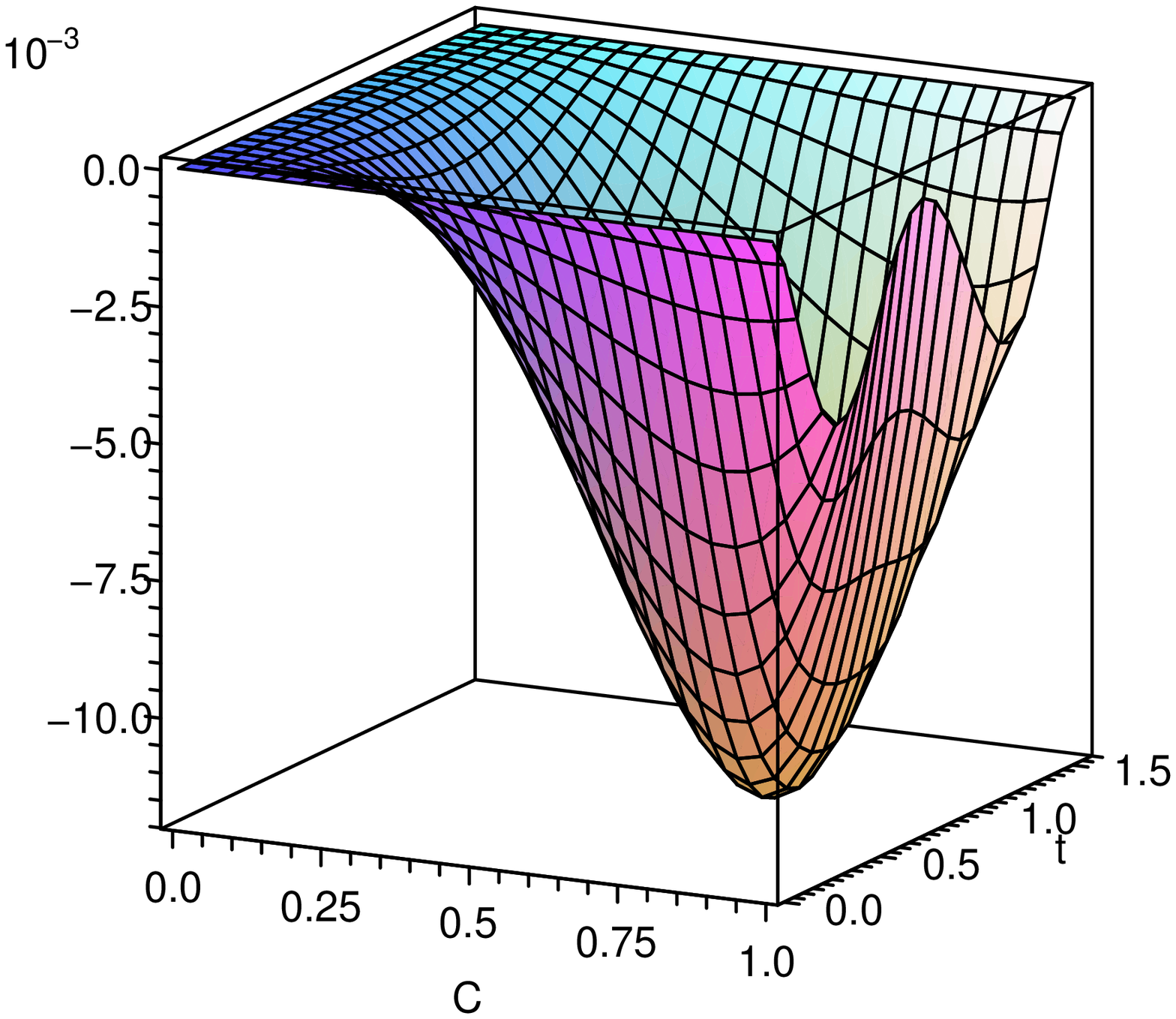}\\
(c) Graph of $h_{1.43}(x,y)$.
\end{center}
} \caption{(a), (b) and (c) illustrate the graphs of
$h_{\alpha}(x,y)$ for $\alpha=1.9$, $1.44$ and $1.43$ respectively.
In each picture, the variables $x$ and $y$ are reparameterized as
$x=C\cos t$ and $y=C\sin t$ for $0\leq C \leq1$ and $0\leq t \leq
\pi/2$. The vertical axes in (a), (b) and (c) are scaled by
$10^{-12}$, $10^{-3}$ and $10^{-3}$ respectively.
 }\label{hgraph2}
\end{figure}

In fact, for the limiting case of $\alpha \rightarrow 1$,
$h_{\alpha}(x,y)$ was analytically shown to be a non-positive
function~\cite{bgk}, that is,
\begin{equation}
{\mathcal E}(\sqrt{x^2+y^2})\leq {\mathcal E}(x)+{\mathcal E}(y),
\label{Epoly}
\end{equation}
for non-negative $x$ and $y$ such that $0 \leq x^2+y^2 \leq
1$ where ${\mathcal E}(x)=\lim_{\alpha\rightarrow 1}
f_{\alpha}(x)$ is defined in Eq.~(\ref{eps}).
Again, due to the continuity of $h_{\alpha}(x,y)$
with respect to $\alpha$, we can assure that
$h_{\alpha}(x,y)$ is nonpositive everywhere when $\alpha$ is around
1. For a specific region of $\alpha$, we have numerically tested for
various values of $\alpha$, and $h_{\alpha}(x,y)$ is observed to
have non-positive function values for $0 \leq \alpha \leq
1.43+\epsilon$ with $0 <\epsilon <0.01$, which is illustrated in
Figure~\ref{hgraph2} (c).

\begin{Con}
For any positive $0 \leq \alpha \leq 1.43+\epsilon$ with some
positive $\epsilon$ such that $0\leq\epsilon<0.01$ and the function
\begin{equation}
f_{\alpha}(x)=\frac{1}{1-\alpha}\log\left[\left(\frac{1-\sqrt{1-x^2}}{2}\right)^{\alpha}
+\left(\frac{1+\sqrt{1-x^2}}{2}\right)^{\alpha}\right],
\end{equation}
defined on the domain $\mathcal{D}=\{(x,y)|0 \leq x,~y,~x^2+y^2 \leq
1\}$, we have
\begin{equation}
f_{\alpha}(\sqrt{x^2+y^2})\leq
f_{\alpha}(x)+f_{\alpha}(y).
\label{fpoly}
\end{equation}
\label{Con: hpositive}
\end{Con}

Now, Conjectures~\ref{Con: alphaless1} and~\ref{Con: hpositive}
lead us to the following theorem, which is the last main result of
this paper.

\begin{Thm}
For a multi-qubit pure state $\ket{\psi}_{A_1 \cdots A_n}$ and any
real $\alpha$ such that $0.83-\epsilon \leq \alpha \leq
1.43+\epsilon$ with $0<\epsilon<0.01$, we have
\begin{equation}
E_{\alpha}\left( \ket{\psi}_{A_1(A_2 \cdots A_n)}\right)\leq  E_{\alpha}^a(\rho_{A_1 A_2}) +\cdots+E_{\alpha}^a(\rho_{A_1
A_n}),
\label{nRdual}
\end{equation}
where $E_{\alpha}\left( \ket{\psi}_{A_1(A_2 \cdots
A_n)}\right)=S_{\alpha}\left(\rho_A \right)$ is the R\'enyi-$\alpha$
entanglement of $\ket{\psi}_{A_1\left(A_2 \cdots A_n\right)}$ with
respect to the bipartite cut $A_1$ and $A_{2}\cdots A_{n}$, and
$E_{\alpha}^a(\rho_{A_1 A_i})$ is the REoA of the reduced density
matrix $\rho_{A_1 A_i}$ for $i=2,\cdots,n$. \label{Rpoly}
\end{Thm}

\begin{proof}
The proof method follows the construction used in~\cite{bgk}. From
the polygamy inequality of entanglement in multi-qubit systems in
Eq.~(\ref{nCdual}) gives us an inequality
\begin{equation}
\mathcal{C}_{A_1 (A_2 \cdots A_n)}  \leq \sqrt{
(\mathcal{C}^a_{A_1 A_2})^2 +\cdots+(\mathcal{C}^a_{A_1 A_n})^2 }.
\label{nCdual2}
\end{equation}

First, let us assume that $ (\mathcal{C}^a_{A_1 A_2})^2
+\cdots+(\mathcal{C}^a_{A_1 A_n})^2 \leq 1$, then we have
\begin{eqnarray}
E_{\alpha}\left( \ket{\psi}_{A_1(A_2 \cdots A_n)}\right)
&=&f_{\alpha}(\mathcal{C}_{A_1 (A_2 \cdots A_n)})\nonumber\\
&\leq &f_{\alpha}\left(\sqrt{
(\mathcal{C}^a_{A_1 A_2})^2 +\cdots+(\mathcal{C}^a_{A_1 A_n})^2 }\right)\nonumber\\
&\leq &f_{\alpha}\left( \mathcal{C}^a_{A_1 A_2}\right) +
f_{\alpha}\left(\sqrt{(\mathcal{C}^a_{A_1 A_3})^2
+\cdots+(\mathcal{C}^a_{A_1 A_n})^2}\right)\nonumber\\
&\leq &f_{\alpha}\left( \mathcal{C}^a_{A_1 A_2}\right)+
f_{\alpha}\left(\mathcal{C}^a_{A_1 A_3}\right) +\cdots + f_{\alpha}\left(
\mathcal{C}^a_{A_1 A_n}\right)\nonumber\\
&\leq& E_{\alpha}^a\left(\rho_{A_1 A_2}\right)
+\cdots+E_{\alpha}^a\left(\rho_{A_1 A_n}\right),
\end{eqnarray}
where the first equality is by the functional relation between the
concurrence and the R\'enyi-$\alpha$ entanglement for $2\otimes d$
pure states, the first inequality is due to the monotonicity of the
function $f_{\alpha}(x)$, the second, third and forth inequalities
are obtained by iterative use of Eq.~(\ref{fpoly}) in
Conjecture~\ref{Con: hpositive}, and the last inequality is by
Eq.~(\ref{f_Ea}).

Now, assume that $ (\mathcal{C}^a_{A_1 A_2})^2
+\cdots+(\mathcal{C}^a_{A_1 A_n})^2 > 1$. Since
$E_{\alpha}\left( \ket{\psi}_{A_1(A_2 \cdots A_n)}\right)=S_{\alpha}(\rho_{A_1})\leq
1$ for any multi-qubit pure state $\ket{\psi}_{A_1(A_2 \cdots A_n)}$, it
is enough to show that $E_{\alpha}^a(\rho_{A_1 A_2}) +\cdots+E_{\alpha}^a(\rho_{A_1
A_n}) \geq 1$. Here, we note that there exist $k \in \{2,\ldots ,n-1 \}$ that
satisfies
\begin{equation}
(\mathcal{C}^a_{A_1 A_2})^2 +\cdots+(\mathcal{C}^a_{A_1 A_k})^2
\leq 1,~ (\mathcal{C}^a_{A_1 A_2})^2 +\cdots+(\mathcal{C}^a_{A_1
A_{k+1}})^2 >1.
\end{equation}
By letting
\begin{equation}
T:=(\mathcal{C}^a_{A_1 A_2})^2 +\cdots+(\mathcal{C}^a_{A_1
A_{k+1}})^2-1,
\end{equation}
we have
\begin{eqnarray}
1 &=& f_{\alpha} \left( \sqrt{(\mathcal{C}^a_{A_1 A_2})^2
+\cdots+(\mathcal{C}^a_{A_1
A_{k+1}})^2-T} \right)\nonumber\\
&\leq& f_{\alpha} \left( \sqrt{(\mathcal{C}^a_{A_1 A_2})^2
+\cdots+(\mathcal{C}^a_{A_1 A_k})^2} \right)
+f_{\alpha} \left( \sqrt{(\mathcal{C}^a_{A_1 A_{k+1}})^2-T} \right)\nonumber\\
&\leq&f_{\alpha} \left( \mathcal{C}^a_{A_1 A_2}\right)+\cdots
+f_{\alpha} \left( \mathcal{C}^a_{A_1 A_k}\right)+
f_{\alpha} ( \mathcal{C}^a_{A_1 A_{k+1}})\nonumber\\
&\leq& E_{\alpha}^a(\rho_{A_1 A_2})+\cdots + E_{\alpha}^a(\rho_{A_1 A_n}),
\end{eqnarray}
where the first inequality is by using Eq.~(\ref{fpoly}) with respect to
$(\mathcal{C}^a_{A_1 A_2})^2+\cdots+(\mathcal{C}^a_{A_1 A_k})^2$ and
$(\mathcal{C}^a_{A_1 A_{k+1}})^2-T$,
the second inequality is by iterative use of Eq.~(\ref{fpoly}) on
$(\mathcal{C}^a_{A_1 A_2})^2+\cdots+(\mathcal{C}^a_{A_1 A_k})^2$, and
the last inequality is by Eq.~(\ref{f_Ea}).
\end{proof}

\section{Conclusion}
\label{Conclusion}

By using R\'enyi-$\alpha$ entropy, we have established a class of
monogamy and polygamy inequalities of multi-qubit entanglement. We
have shown that monogamy of multi-qubit entanglement can have
CKW-type characterization in terms of R\'enyi-$\alpha$ entanglement
for $\alpha \geq2$, and conjectured the possible region of $\alpha$
for polygamy inequality in terms of REoA. Although the specific
region of $\alpha$ for polygamy inequality is supported by numerical
evidences, the existence of such region is based on the continuity
of R\'enyi-$\alpha$ entropy, which is analytically provable. Thus,
the conjecture we claimed here is highly reasonable because it has
both numerical and analytical reasons.

Multipartite entanglement is known to have many inequivalent
classes, which are not convertible to each other under {\em
Stochastic Local operations and classical communications}
(SLOCC)~\cite{DVC}. Furthermore, the number of inequivalent classes
increases dramatically as the number of parties increase~\cite{os}.
Not like bipartite entanglement, the existence of inequivalent
classes of multipartite entanglement implies that the states from
different classes are hardly comparable to each other in such a way
of comparing a single parameter that quantifies their entanglement.
This is one of the main difficulties in the study of multipartite
entanglement.

Whereas the interconvertibility of quantum states under SLOCC gives
us an operational way to classify multipartite entanglement,
entangled states from different classes can also reveal different
characters with respect to their monogamy and polygamy properties.
For example, three-qubit systems are known to have two inequivalent
classes of genuine three-qubit entanglement, the
Greenberger-Horne-Zeilinger (GHZ) class~\cite{GHZ} and the
W-class~\cite{DVC}. In terms of monogamy and polygamy relations, CKW
and its dual inequalities are saturated by W-class states, while the
differences between terms in the inequalities can assume their
largest values for GHZ-class states. In other words, monogamy and
polygamy of multipartite entanglement can also be used for an
analytical characterization of entanglement in multipartite quantum
systems.

The class of monogamy and polygamy inequalities of multi-qubit
entanglement we provided here consists of infinitely many
inequalities parameterized by $\alpha$, and each inequality is based
on the distinct character of R\`enyi-$\alpha$ entropy with respect
to $\alpha$. We believe that this selective choice of our monogamy
and polygamy inequalities will leads us to an efficient way of
analytic classification of multi-qubit entanglement. Moreover, our
result will also provide useful tools and strong candidates for
general monogamy and polygamy relations of entanglement in
multipartite higher-dimensional quantum systems, which is one of the
most important and necessary topics in the study of multipartite
quantum entanglement.

\section*{Acknowledgments}
This work was supported by {\it i}CORE, MITACS, and USARO. BSC is a CIFAR Associate.

\end{document}